%% file: main.tex
\documentclass[11pt]{article}
\usepackage{lmodern} 
\usepackage[T1]{fontenc}

\usepackage{color}
\usepackage[colorlinks=true,linkcolor=blue,citecolor=blue]{hyperref}
\usepackage{amsmath, amssymb, amsthm}
\usepackage{mathtools}
\usepackage[margin=1in]{geometry}
\usepackage{graphics}
\usepackage{pifont}
\usepackage{tikz}
\usepackage{bbm, bm}

\usetikzlibrary{arrows.meta}
\usepackage{environ}
\usepackage{framed}
\usepackage{url}
\usepackage[linesnumbered,ruled,vlined]{algorithm2e}
\usepackage[noend]{algpseudocode}
\usepackage[labelfont=bf]{caption}
\usepackage{framed}
\usepackage[framemethod=tikz]{mdframed}
\usepackage{appendix}
\usepackage{graphicx}
\usepackage[textsize=tiny]{todonotes}
\usepackage{tcolorbox}
\allowdisplaybreaks[1]
\usepackage{nicefrac}
\usepackage{thm-restate}
\usepackage[noabbrev,capitalize,nameinlink]{cleveref}
\crefname{equation}{}{}

\usepackage{subcaption}

\DontPrintSemicolon
\SetKw{KwAnd}{and}
\SetProcNameSty{textsc}
\SetFuncSty{textsc}

\newcommand\remove[1]{}

\newtheorem{lemma}{Lemma}[section]
\newtheorem*{lemma*}{Lemma}
\newtheorem{theorem}[lemma]{Theorem}
\newtheorem{corollary}[lemma]{Corollary}
\newtheorem*{corollary*}{Corollary}

\newtheorem*{theorem*}{Theorem}
\newtheorem*{inducthyp*}{Inductive Hypothesis}
\newtheorem*{definition*}{Definition}

\newtheorem*{rem*}{Remark}

\newcommand\R{\mathbb{R}}

\newcommand{\eps}{\varepsilon}
\renewcommand{\O}{\widetilde{O}}

\newcommand{\assign}{\leftarrow}

\renewcommand{\forall}{\mathrm{\text{ for all }}}

\newcommand{\vone}{\mathbf{1}}

\newcommand{\bc}{\bm{c}}
\newcommand{\bd}{\boldsymbol{d}}
\renewcommand{\bf}{\bm{f}}

\newcommand{\br}{\boldsymbol{r}}

\newcommand{\bu}{\boldsymbol{u}}

\newcommand{\bDelta}{\boldsymbol{\Delta}}

\newcommand{\diag}{\mathrm{diag}}

\renewcommand{\hat}{\widehat}

\DeclareFontFamily{U}{mathb}{\hyphenchar\font45}
\DeclareFontShape{U}{mathb}{m}{n}{<5> <6> <7> <8> <9> <10> gen * mathb
<10.95> mathb10 <12> <14.4> <17.28> <20.74> <24.88> mathb12}{}
\DeclareSymbolFont{mathb}{U}{mathb}{m}{n}
\DeclareMathSymbol{\rcirclearrow}{\mathbin}{mathb}{'367}

\renewcommand{\bar}{\overline}

\foreach \x in {A,...,Z}{%
	\expandafter\xdef\csname m\x\endcsname{\noexpand\mathbf{\x}}
}

\foreach \x in {A,...,Z}{%
	\expandafter\xdef\csname c\x\endcsname{\noexpand\mathcal{\x}}
}

\newif\ifrandom
\randomtrue

\newcommand{\defeq}{\stackrel{\mathrm{\scriptscriptstyle def}}{=}}

\def\norm#1{\left\| #1 \right\|}

\renewcommand{\bar}[1]{\overline{#1}}

\newcommand{\va}{\boldsymbol{a}}
\newcommand{\vb}{\boldsymbol{b}}
\newcommand{\vc}{\boldsymbol{c}}
\newcommand{\vr}{\boldsymbol{r}}

\newcommand{\new}{\mathrm{new}}
\renewcommand{\hat}{\widehat}

\newcommand{\poly}{{\mathrm{poly}}}

\newcommand{\todolater}[1]{}

\DeclareUnicodeCharacter{2113}{$\ell$}

\usepackage[backend=biber, isbn=false, style=alphabetic, backref=true, doi=false, url=false, maxcitenames=10, mincitenames=5, maxalphanames=10, maxbibnames=10, minbibnames=5, minalphanames=5, defernumbers=true, sortlocale=en_US]{biblatex}
\begin{document}

\title{Exponential Convergence of Sinkhorn Under Regularization Scheduling}
\author{
Jingbang Chen\\ University of Waterloo\\ j293chen@uwaterloo.ca
\and
Li Chen\\ Georgia Tech\\ lichen@gatech.edu
\and
Yang P. Liu \\ Stanford University \\ yangpliu@stanford.edu
\and
Richard Peng\\ University of Waterloo\\ y5peng@uwaterloo.ca
\and
Arvind Ramaswami \\ Georgia Institute of Technology \\ aramaswami32@gatech.edu
}

\maketitle

\begin{abstract}
In 2013, Cuturi \cite{Cut13} introduced the \textsc{Sinkhorn} algorithm for matrix scaling as a method to compute solutions to regularized \emph{optimal transport} problems. In this paper, aiming at a better convergence rate for a high accuracy solution, we work on understanding the \textsc{Sinkhorn} algorithm under regularization scheduling, and thus modify it with a mechanism that adaptively doubles the regularization parameter $\eta$ periodically. We prove that such modified version of \textsc{Sinkhorn} has an exponential convergence rate as iteration complexity depending on $\log(1/\eps)$ instead of $\eps^{-O(1)}$ from previous analyses \cite{Cut13,altschuler2017near} in the \emph{optimal transport} problems with integral supply and demand. Furthermore, with cost and capacity scaling procedures, the general \emph{optimal transport} problem can be solved with a logarithmic dependence on $1/\eps$ as well.

\end{abstract}

\section{Introduction}


The \emph{optimal transport} (OT) problem asks to compute the minimum cost needed to send supplies to demands. It is formally described as the following linear program:
\begin{align}
\label{eq:OT}
OPT &\defeq \min_{\mX \in \mU(\vr,\vc)}\sum_{i\in[n],j\in[m]} \mQ_{ij}\mX_{ij},~\mU(\vr, \vc) \defeq \left\{\mX \in \mathbb{R}_{+}^{n \times m}: \mX \mathbf{1}_{m}=\vr \quad \text { and } \quad \mX^\top \mathbf{1}_n =\vc\right\}
\end{align}
where $\mQ$ is the given cost matrix, and $\br \in \R^n_+$ and $\bc \in \R^m_+$ are the demand and supply vectors.
In this paper, we want to understand the time complexity of the algorithm for finding a feasible solution $\mX$ whose cost is within $OPT + \eps.$ The \emph{optimal transport} problem is widely used in machine learning, particularly in areas such as computer vision~\cite{damodaran2018deepjdot,Kolkin_2019_CVPR}, natural language processing~\cite{kusner2015word}, deep learning~\cite{oh2020unpaired,zhang2021optimal}, clustering~\cite{ho2017multilevel}, unsupervised learning~\cite{arjovsky2017wasserstein}, and semi-supervised learning~\cite{solomon2014wasserstein}.



In 1964, Richard Sinkhorn discovered that for any positive square matrix $\mA$, there exists a unique doubly stochastic matrix of the form $\mX=\diag(\va) \mA \diag(\vb)$ where $\diag(\va)$ and $\diag(\vb)$ are diagonal matrices with positive entries \cite{sinkhorn1964relationship}.
$\mX$ can be computed using the \textsc{Sinkhorn} algorithm. This algorithm normalizes the rows and columns of the matrix in an alternating fashion~\cite{sinkhorn1967concerning}.
In 2013, Cuturi showed that the matrix scaling method can be used to approximate solutions to the \emph{optimal transport} problem with regularization~\cite{Cut13}. Such regularization is achieved by adding an entropy regularizer $\eta^{-1} \sum_{i \in [n], j \in [m]} \mX_{ij}(\log \mX_{ij} - 1)$ to the OT objective function. The idea of solving regularized OT was already introduced in 1980s under the name of gravity models \cite{peyre2019computational}.


The convergence rate of the \textsc{Sinkhorn} algorithm has been the subject of both theoretical and practical analyses in various settings.
For instance, it has been proven to have a $\log(1/\eps)$ convergence bound under the Hilbert projective metric~\cite{franklin1989scaling}.
Since the work of \cite{Cut13}, several OT algorithms have been developed using the idea of entropic regularization, which have been efficient in practice~\cite{benamou2015iterative,genevay2016stochastic}.
However, there are only a few theoretical guarantees for the optimal transport problem directly.
\cite{altschuler2017near} shows that with the appropriate choice of parameters, the standard \textsc{Sinkhorn} or \textsc{Greenkhorn} algorithm is a near-linear time approximation algorithm for input data of $n$ dimensions, taking $O(n^2 ||\mQ||_{\infty}^3 (\log n)\eps^{-3})$ runtime to give a solution within $OPT + \eps$.
However, the convergence rate may be significantly slower when seeking high-accuracy solutions due to the $\eps^{-3}$ factor.

To improve the convergence rate in high-accuracy scenarios, we focus on the selection of the regularization parameter $\eta$, which balances the desired accuracy and the iteration complexity of the subroutine.
One approach uses a series of $\{\eta_k\}_{k \geq 1}$ instead of a single value.
In 2019, Bernhard Schmitzer discussed such scheduling in \cite{schmitzer2019stabilized}, providing a new analysis of the \textsc{Sinkhorn} algorithm with regularization scheduling.
In our work, we examine the \textsc{Sinkhorn} algorithm under this scheduling and explore incorporating it into an adaptive regularization scheme.

\subsection{Our Results} 

In this paper, we show that the \textsc{Sinkhorn} algorithm with regularization scheduling has an exponential convergence rate.
This means that the number of iterations needed to achieve an $\eps$-additive error desired is $\poly\log(1/\eps).$
Additionally, the algorithm has a runtime of $\poly(n, m, \log(1/\eps))$ using row/column scaling operations.
The closest similar result to this is the weakly polynomial time matrix scaling algorithm in \cite{LSW98}, which uses a more complicated scaling procedure. We provide a table comparing our result with some previous works in Table \ref{tab:sinkhorn}.

\renewcommand{\arraystretch}{2}
\begin{table}[h!]
    \centering
    \begin{tabular}{|c|c|c|} 
      \hline
      Algorithm & \# of Iterations & Comments \\
      \hline
      \Cref{theo:algo} & $\O\left(\|\br\|_1^2 \log(\|\mQ\|_{\infty} / \eps)\right)$ & Integral OT \\
      \hline
      \Cref{theo:costCapScaling} & $\poly(n, m, \log(1 / \eps), \log ||\mQ||_{\infty}, \log ||\vr||_1)$ & General OT \\
      \hline
      \cite{altschuler2017near} & $\O(\|\mQ\|_\infty^3 / \eps^3)$ & Plain \textsc{Sinkhorn} with $\eta = \log n / \eps$ \\
      \hline
      \cite{franklin1989scaling} & $O(\exp(\|\mQ\|_\infty \log n / \eps) \log (1 / \eps))$ & Plain \textsc{Sinkhorn} with $\eta = \log n / \eps$ \\
      \hline
      \cite{LSW98} & $\O(n^5 \log(1 / \eps))$ & Modified row/column scaling \\
      \hline
    \end{tabular}
    \caption{\textsc{Sinkhorn}-based algorithms}
    \label{tab:sinkhorn}
\end{table}

For the analysis, we first focus on cases where the demands and supplies are integers bounded by some integer $\mu$. The convergence result is summarized as follows:
\begin{theorem}[Algorithmic result]
\label{theo:algo}
If both the demand vector $\br$ and the supply vector $\bc$ are integral and bounded by $\mu$, i.e. $\mu = \max\{\|\br\|_\infty, \|\bc\|_\infty\}$, 
\Cref{algo:scaling} computes a feasible solution $\mX$ to \eqref{eq:OT} with $\eps$-additive error using
\begin{align*}
O\left(\norm{\br}_1^2 \log\left(n \mu\right) \log\left(\norm{\mQ}_{\infty} \norm{\br}_1 / \eps\right)\right)
\end{align*}
iterations of row/column scaling operations.
\end{theorem}

Additionally, note that if $\vr$ is integer and $||\vr||_1 = O(n)$ (which is relevant in problems like weighted bipartite matching), then Theorem 1.1 gives a stronger guarantee than \cite{LSW98}.

We will provide a detailed explanation and proof of our statement in \Cref{sec2}.
Essentially, the proof is based on analyzing the duality gap of the regularized optimal transport problem.
Given a good primal-dual solution pair, we show that after doubling the regularization parameter, the duality gap is proportional to $1 / \eta.$
On the other hand, we also show that a row/column scaling operation reduces the duality gap by roughly $1 / \eta$.
Both $1 / \eta$ terms cancel each other and we can efficiently find a good primal-dual pair w.r.t. to the doubled $\eta$.

To achieve $\poly(n, m, \log(1/\eps))$ runtime and to handle non-integral input, we use a cost/capacity scaling scheme commonly used in network flow algorithms (see Appendix C in \cite{CKLPPGS22}).
The method involves reducing \eqref{eq:OT} to $O(\log(\norm{\mQ}_\infty) \log(\mu))$ instances each with a dimension of at most $2n^2$ and demand/supply entries at most $n^8$.


To handle fractional input, we can round each cost, demand, and supply entry to the nearest integral multiple of $\poly(\eps, 1/n, 1/m)$, that is, an integral instance with $\mu = \max\{\|\br\|_\infty, \|\bc\|_\infty\} \cdot \poly(n, m, 1/\eps)$.
This allows us to solve the problem in $\poly(n, m, \log(1/\eps))$ time.
However, this solution may not be feasible for the original fractional input.
But, we can use standard rounding methods to make the solution feasible such as Algorithm 2 in \cite{altschuler2017near}. This process is summarized in the following Lemma.


\begin{theorem}[Polynomial Runtime via Cost/Capacity Scaling and Rounding]
\label{theo:costCapScaling}
There is an algorithm that gives a solution $\mX$ to \eqref{eq:OT} with $\eps$ additive error with $O(\log(\norm{\mQ}_\infty) \log(\mu))$ calls to Algorithm~\ref{algo:scaling} on integral OT instances with dimension at most $n^2$ and the total demand/supply at most $O(n^{10}).$
\end{theorem}



\subsection{Related Work} 
\paragraph{Optimal Transport} 
Many combinatorial techniques have been introduced to compute the exact solution for certain kinds of OT problems. The Hungarian method invented by Kuhn \cite{kuhn1955hungarian} in $1955$ solves the assignment problem (equivalent to OT) in $O(n^3)$ time. In $1991$, Gabow and Tarjan gave an $O(n^{2.5} \log(nN))$ time cost/capacity scaling algorithm \cite{GR91} to solve OT, where $N$ is the largest element in the scaled cost matrix. Using cost/capacity scaling techniques, min-cost flow algorithms such as network simplex also provide exact algorithms for the optimal transport problem in $O(n^3 \log n \log(nN))$ time \cite{OTML19}. There are also studies on certain kinds of OT problems, such as geometric OT \cite{alvarez2020geometric} \cite{peyre2019computational}. Additionally, there has been significant recent theoretical work studying the runtime of solving mincost flow, which generalizes OT \cite{LS19,matching,mincost,CKLPPGS22}. These methods rely heavily on second order methods and primitives from graph theory.

\paragraph{Regularization} In machine learning, regularization is widely used to resolve various kinds of datasets' heterogeneity \cite{tian2022comprehensive,zhu2018machine,neyshabur2017implicit,goodfellow2016regularization}. Recently, there have been more works on developing adaptive regularization methods, including deep learning on imbalanced data \cite{cao2020heteroskedastic} and learning neural networks \cite{NEURIPS2019_2281f5c8}. There are also studies on regularization hyperparameter selection \cite{luketina2016scalable,leung1999adaptive}. 


\subsection{Notation} We use bold lowercase characters such as $\va$ to denote vectors. Specially, we use $\vone$ or $\vone_n$ to denote the all ones vector with proper length. We use bold capital letters (such as $\mQ$) as matrices. Specially, we denote the matrix that we are rescaling as $\mX$. We denote the inner product of two matrix as $\langle \cdot, \cdot \rangle$, so $\langle \mX, \mQ \rangle = \sum_{i\in[n],j\in[m]} \mX_{ij}\mQ_{ij}$. We use the integral vectors $\vr \in \mathbb{Z}^n$ and $\vc \in \mathbb{Z}^m$ to denote the desired row and column sums. Note that the matrix $\mX$ has row sums $\mX \vone$ and column sums $\mX^\top \vone$. We use $\alpha_i$ for $i \in [n]$ and $\beta_j$ for $j \in [m]$ as the dual variables in our matrix scaling algorithm. As above, $\eta$ is the regularization parameter.

\section{Matrix Scaling with Regularization Scheduling}
\label{sec2}

We propose an algorithm \textsc{ExpSinkhorn} to solve the OT problem to high accuracy. The algorithm maintains a matrix $\mX$ to be scaled and a regularization parameter $\eta$. It rescales the rows and columns iteratively for this fixed parameter $\eta$. When the rows and columns are close enough to scaled, the algorithm doubles $\eta$. We ultimately show that this algorithm converges in time depending logarithmically on $\eps^{-1}$ (see Theorem \ref{theo:algo}), as opposed to the standard Sinkhorn algorithm requiring time depending polynomially on $\eps^{-1}$ to converge \cite{Cut13,altschuler2017near}.

The analysis of our algorithm hinges on understanding the interaction between the $\ell_1$ error of the row/column scaling and a \emph{dual objective}. Formally, when the quantities $ \|\mX\vone-\vr\|_1 \enspace \text{ and } \enspace \|\mX^\top\vone - \vc\|_{1} $
are small, the algorithm doubles the regularization parameter $\eta$. We show that when they are large, then rescaling the rows or columns of $\mX$ causes the \emph{dual objective} to significantly improve (see Lemma \ref{lemma:dualincrease}). We also prove that when the $\ell_1$ errors are small, the duality gap is small (see Lemma \ref{lemma:dualitygap}), which bounds the number of iterations (see Lemma \ref{lemma:itercount}).

We now formally present our matrix scaling algorithm that doubles $\eta$ over time to give a high accuracy solution to optimal transport.

\begin{algorithm}[h]
\caption{\textsc{ExpSinkhorn}$(\mQ, \vr, \vc, \epsilon)$ - Solves the optimal transport problem.}
\label{algo:scaling}
    \KwIn{A $n \times m$ cost matrix $\mQ$.}
    \KwOut{A $n \times m$ matrix $\mX_{i j} \geq 0$ such that $\mX \in \mU(\vr,\vc)$ and $\langle \mX, \mQ \rangle \le OPT + \epsilon$, where%
\[ 
OPT \defeq \min_{\mX_{i j} \geq 0, \mX \in \mU(\vr, \vc)} \langle \mX, \mQ \rangle. \]}

\DontPrintSemicolon
\everypar={\nl} 
$\mu \assign \max\{\max_{i \in [n]} r_i, \max_{j\in[m]} c_j\}.$ \par
$\vr \assign \vr/\mu, \vc \assign \vc/\mu$ \Comment{Scale $\vr, \vc$ to have $\|\vr\|_\infty \le 1, \|\vc\|_\infty \le 1$.} \label{line:scaler} \par
$\eta \assign 10\|\mQ\|_{\infty}^{-1} \log(n\mu).$ \Comment{Starting value of $\eta$.} \par
$\alpha_i \assign -\|\mQ\|_\infty, \beta_j \assign -\|\mQ\|_\infty$ for $i\in[n]$, $j\in[m]$. \Comment{Dual variable initialization} \par
\While{$\eta \le 4\mu\eps^{-1}\|\vr\|_1\log(n\mu)$}{ \label{line:mainwhile}
    $\mX_{ij} \assign \exp(\eta(\alpha_i+\beta_j-\mQ_{ij})).$ \Comment{Initialize matrix to be scaled.} \par
    \For{$k \ge 0$}{
        $\va \assign \mX\vone.$ \Comment{Row sums} \par
        $\vb \assign \mX^T\vone.$ \Comment{Column sums} \par
        
        \If{$\|\va-\vr\|_1 > 1/(2\mu)$}{\label{line:ifrow}
            $\mX_{ij} \assign (\va_i/\vr_i)^{-1}\mX_{ij}$ for $1 \le i \le n$, $1 \le j \le m$ \Comment{Row scaling} \label{line:rescale} \par
            $\alpha_i \assign \alpha_i - \eta^{-1}\log(\va_i/\vr_i)$ for $1 \le i \le n$ \label{line:adjustr} \Comment{Row dual adjustment} \par
        }\uElseIf{$\|\vb-\vc\|_1 > 1/(2\mu)$}{\label{line:ifcol}
            $\mX_{ij} \assign (\vb_j/\vc_j)^{-1}\mX_{ij}$ for $1 \le i \le n$, $1 \le j \le m$ \label{line:rescale2} \Comment{Column scaling}\par
            $\beta_j \assign \beta_j - \eta^{-1}\log(\vb_j/\vc_j)$ for $1 \le j \le m$ \Comment{Row dual adjustment} \label{line:adjustc}
        } \Else {
            {\label{line:rif}
                $\eta \assign 2\eta$ and return to line \ref{line:mainwhile}. \label{line:double} \par
            }
        }
    }
}
$\mX \assign \mu\mX$ \Comment{Scale $\mX$ back up.} \par
Repair the demands routed by $\mX$ and return $\mX$.
\end{algorithm}
We will assume throughout this analysis that $\|\vr\|_\infty, \|\vc\|_\infty \le 1$, such that $\mu\vr, \mu\vc \in \mathbb{Z}^n$. This is because we scale $\vr, \vc$, which are originally in $\mathbb{Z}^n$, down by $\mu$ in line \ref{line:scaler} of Algorithm \ref{algo:scaling}.

The analysis is based on looking at the dual program of the optimal transport objective: \[ \max_{\alpha_i+\beta_j \le \mQ_{ij} \forall i \in [n], j \in [m]} \sum_{i\in[n]} \vr_i\alpha_i + \sum_{j\in[m]} \vc_j\beta_j. \]
The value of this program is also $OPT$, the same as the value of the optimal transport objective $\min_{\mX \ge 0, \mX \in \mU(\vr, \vc)} \langle \mX, \mQ \rangle$ by linear programming duality.

Thus, as long as we can guarantee that the $\alpha_i, \beta_j$ parameters in Algorithm \ref{algo:scaling} always satisfy $\alpha_i + \beta_j \le \mQ_{ij}$, then the dual potential $D := \sum_{i\in[n]} \vr_i\alpha_i + \sum_{j\in[m]} \vc_j\beta_j \le OPT$ at all times. We will show these by induction.
\begin{lemma}[Algorithm invariants]
\label{lemma:valid}
At all times during an execution of Algorithm \ref{algo:scaling}, we have that $\mX_{ij} \le 1$ for all $i \in [n], j \in [m]$ and $\sum_{i, j} \mX_{ij} \le \|\vr\|_1$. Hence $\alpha_i + \beta_j \le \mQ_{ij}$ at all times.
\end{lemma}
\begin{proof}
The ``hence'' part follows because $\mX_{ij} = \exp(\eta(\alpha_i + \beta_j - \mQ_{ij}))$, so if $\mX_{ij} \le 1$ then $\alpha_i + \beta_j \le \mQ_{ij}$. Thus, in this proof we focus on showing the claims about $\mX$.

We will proceed by induction. We first check that all conditions hold at the start of the algorithm. For the initial choices of $\eta, \alpha_i, \beta_j$ we have that
\[ \mX_{ij} = \exp(10\|\mQ\|_\infty^{-1} \log n(-2\|\mQ\|_\infty + \mQ_{ij})) \le \exp(-10\log n) \le n^{-10}. \]
Hence, $\sum_{ij} \mX_{ij} \le n^{-8} \le 1 \le \|\vr\|_1$, and $\mX_{ij} \le 1$ for all $i, j$.

Now, we check that the condition continues to hold after we double $\eta$ in line \ref{line:double}. Let $\mX^{\new}$ be the new matrix after $\eta$ is doubled. Clearly, $\mX^{\new}_{ij} = \mX_{ij}^2 \le \mX_{ij}$ because $\mX_{ij} \le 1$ by induction. So $\sum_{ij} \mX^{\new}_{ij} \le \sum_{ij} \mX_{ij} \le \|\vr\|_1$ by induction, and $\mX^{\new}_{ij} = \mX_{ij}^2 \le \mX_{ij} \le 1$.

Finally, we check the conditions after a rescaling step in lines \ref{line:rescale}, \ref{line:rescale2}. By symmetry, we consider a row rescaling step in line \ref{line:rescale}. After such a step, we know that $\sum_{j \in [m]} \mX_{ij} = \vr_i$ for all $i \in [n]$. Because $\|\vr\|_\infty \le 1$, we deduce that $\mX_{ij} \le 1$ for all $i, j$ and $\sum_{ij} \mX_{ij} \le \|\vr\|_1$ as desired. The same argument applies to a column rescaling in line \ref{line:rescale2}, if we note that $\|\vc\|_1 = \|\vr\|_1$.
\end{proof}
Because the dual potentials $\alpha_i, \beta_j$ are feasible, we know that the dual potential is upper bounded.
\begin{corollary}[Dual potential upper bound]
\label{cor:valid}
During an execution of Algorithm \ref{algo:scaling}, $\alpha_i, \beta_j$ satisfy $D := \sum_{i\in[n]} \vr_i\alpha_i + \sum_{j\in[m]} \vc_j\beta_j \le OPT$ at all times.
\end{corollary}
\begin{proof}
By Lemma \ref{lemma:valid} we know that $\alpha_i + \beta_j \le \mQ_{ij}$ at all times. As noted above, by linear programming duality \[ D \le \max_{\alpha_i + \beta_j \le \mQ_{ij} \forall i\in[n],j\in[m]} \sum_{i\in[n]} \vr_i\alpha_i + \sum_{j\in[m]} \vc_j\beta_j = OPT. \]
\end{proof}
The remainder of the analysis requires the following claims. First, we show that the duality gap $OPT - D$ is small when $\|\vr-\va\|_1 \le 1/(2\mu)$ and $\|\vc-\vb\|_1 \le 1/(2\mu)$ trigger, i.e. line \ref{line:double}. When these do not hold, we show that a rescaling step in lines \ref{line:adjustr} or \ref{line:adjustc} causes $D$ to significantly increase. Finally, we will show how to round our approximately scaled solution $\mX$ to a feasible point.

Towards this, we show the following useful helper lemma which intuitively shows that an approximately feasible $\mX$ ``contains'' half of a truly feasible solution.
\begin{lemma}[Containing a feasible solution]
\label{lemma:hatmx}
Let $\vr, \vc$ be vectors with $\|\vr\|_1, \|\vc\|_1 \le 1$ and $\mu\vr, \mu\vc \in \mathbb{Z}^n$. If $\mX \ge 0$ satisfies $\mX\vone = \vr$ and $\|\mX^\top\vone - \vc\|_1 \le 1/(2\mu)$, then there is a vector $\hat{\mX} \in \R^{n \times m}$ with $0 \le \hat{\mX}_{ij} \le \mX_{ij}$ for all $i \in [n], j \in [m]$ and $\hat{\mX}\vone = \vr/2$ and $\hat{\mX}^\top\vone = \vc/2$.

Additionally, such an $\hat{\mX}$ can be found by running any maximum flow algorithm.
\end{lemma}
Clearly we may swap the roles of $\vr, \vc$ above. We state only one case in Lemma \ref{lemma:hatmx} for brevity.
\begin{proof}
Let $\alpha \ge 0$ be maximal so that there exists a $0 \le \hat{\mX} \le \mX$ such that $\hat{\mX}\vone=\alpha\vr$ and $\hat{\mX}^\top\vone=\alpha\vc$. Let $\mY$ satisfy $\mY\vone=\alpha\vr$ and $\mY^\top\vone=\alpha\vc$. We wish to show that $\alpha\ge1/2$.

Assume $\alpha < 1/2$ for contradiction, and let $\mX^{(1)} = \mX - \mY$, so that $\mX^{(1)}\vone = (1-\alpha)\vr$ and $\|(\mX^{(1)})^\top\vone - (1-\alpha)\vc\|_1 \le 1/(2\mu)$. Multiplying the previous equations by $(1-\alpha)^{-1}\mu$ on both sides yields that
\begin{align}
    \bar{\mX}\vone = \mu\vr \enspace \text{ and } \enspace \left\|\bar{\mX}^\top\vone - \mu\vc\right\|_1 \le \frac{1}{2(1-\alpha)} < 1, \label{eq:scaledup}
\end{align}
where $\bar{\mX} := (1-\alpha)^{-1}\mu\mX^{(1)}$. Note that if there exists $0 \le \mZ \le \bar{\mX}$ such that $0 \le \mZ \le \bar{\mX}$
and $\delta > 0$ with $\mZ\vone = \delta\vr$ and $\mZ^T\vone = \delta\vc$, then letting $\mW = (1 - \alpha)\mu^{-1}\mZ + \mY$ gives that $\mW \leq \mX^{(1)} + \mY \leq \mX$, and $\mW \vone = (\alpha + (1-\alpha)\mu^{-1}\delta)\vr, \mW^T\vone = (\alpha + (1-\alpha)\mu^{-1}\delta)\vc$, contradicting the maximality of $\alpha$. Thus, it suffices to use the fact that both $\mu\vr$ and $\mu\vc$ are integral vectors to construct $0 \le \mZ \le \bar{\mX}$ and $\delta > 0$ such that $\mZ\vone = \delta\vr$ and $\mZ^T\vone = \delta\vc$.

Let $E$ be the support of $\bar{\mX}$, i.e. $E := \left\{(i, j) : \bar{\mX}_{ij} > 0 \right\}$. For a subset $S \subseteq [n]$, let $N(S) := \{t : \exists s \in S, (s, t) \in E\}$, i.e. the neighborhood of $S$. By Hall's marriage theorem (for weighted sources and sinks), the subset $E$ supports a flow between $\mu\vr$ and $\mu\vc$ as long as for all subsets $S \subseteq [n]$, we have that $\sum_{s \in S} (\mu\vr)_s \le \sum_{t \in N(S)} (\mu\vc)_t$. By the guarantee in (\ref{eq:scaledup}) we know that
\begin{align*}
\sum_{s \in S} (\mu\vr)_s &= \sum_{(s, t) \in E} \bar{\mX}_{st} \le \sum_{t \in N(S)} \sum_{s \in [n]} \bar{\mX}_{st} \\ &\le \sum_{t \in N(S)} (\mu\vc)_t + \|\bar{\mX}^\top\vone - \mu\vc\|_1 < \sum_{t \in N(S)} (\mu\vc)_t + 1.
\end{align*}
Because $\sum_{s \in S} (\mu\vr)_s$ and $\sum_{t \in N(S)} (\mu\vc)_t$ are both integral quantities, the previous equation implies that $\sum_{s \in S} (\mu\vr)_s \le \sum_{t \in N(S)} (\mu\vc)_t$ as desired. This shows that there is some $0 \le \mZ \le \bar{\mX}$ and strictly positive $\delta > 0$ such that $\mZ\vone = \delta\vr$ and $\mZ^T\vone = \delta\vc$. This completes the proof.
\end{proof}
The above lemma lets us bound the duality gap right before we double $\eta$, i.e. when line \ref{line:double} occurs.
\begin{lemma}[Duality gap]
\label{lemma:dualitygap}
Let $D = \sum_{i \in [n]} \alpha_i\vr_i + \sum_{j \in [m]} \beta_j\vc_j$. During an execution of Algorithm \ref{algo:scaling} when line \ref{line:double} occurs, we have that $OPT - D \le 2\eta^{-1}\|\vr\|_1 \log(n\mu)$.
\end{lemma}
\begin{proof}
We only handle the case where $\mX^\top\vone=\vr$, as the other case is symmetric (recall that $\|\vr\|_1 = \|\vc\|_1$). Hence $\mX^\top\vone=\vr$.

By Jensen's inequality, we know that
\begin{align*}
\sum_{i\in[n],j\in[m]} \mX_{ij} \log \mX_{ij} &= -\sum_{i \in [n]} \vr_i \sum_{j \in [m]} \frac{\mX_{ij}}{\vr_i} \log(1/\mX_{ij}) \\ &\ge -\sum_{i \in [n]} \vr_i \log \left( \sum_{j \in [m]} \frac{\mX_{ij}}{\vr_i}\frac{1}{\mX_{ij}}\right) \\ &= \ge -\sum_{i \in [n]} \vr_i \log(m/\vr_i) \ge -\|\vr\|_1 \log(n\mu),
\end{align*}
because $\vr_i \ge \mu^{-1}$ for all $i$, because $\mu\vr \in \mathbb{Z}^n$ by assumption. Let $\hat{\mX}$ be as constructed in Lemma \ref{lemma:hatmx}. Because $\mX_{ij} \le 1$ for all $i, j$ by Lemma \ref{lemma:valid} (so $\log \mX_{ij} \le 0$), we can write
\begin{align*}
    \sum_{i\in[n],j\in[m]} \mX_{ij} \log \mX_{ij} &\le \sum_{i\in[n],j\in[m]} \hat{\mX}_{ij} \log \mX_{ij} = \eta\sum_{i\in[n],j\in[m]} \hat{\mX}_{ij}(\alpha_i + \beta_j - \mQ_{ij}) \\
    &= \eta(D/2 - \langle \hat{\mX}, \mQ \rangle) \le \eta(D/2 - OPT/2),
\end{align*}
where the final inequality follows because $\hat{\mX}\vone = \vr/2$ and $\hat{\mX}^\top\vone = \vc/2$, hence $\langle \hat{\mX}, \mQ \rangle \ge OPT/2$ by the minimality of OPT.
Combining the previous two expressions completes the proof.
\end{proof}
Now, we prove that if line \ref{line:double} does not occur, then the dual solution increases significantly.
\begin{lemma}[Dual increase]
\label{lemma:dualincrease}
Let $\va = \mX\vone$, and consider updating $\alpha$ as in line \ref{line:adjustr}. Then the dual $D := \sum_{i \in [n]} \alpha_i \vr_i + \sum_{j \in [m]} \beta_j \vc_j$ increases by at least \[ \eta^{-1}/10 \cdot \min\{\mu^{-1}, \|\vr\|_1^{-1}\|\va - \vr\|_1^2\}. \]
\end{lemma}
\begin{proof}
Note the following numerical bound: $-\log(1-t) \ge t + \min\{1/10, t^2/3\}$ for all $t < 1$.
By the formula in line \ref{line:adjustr}, the dual increases by
\begin{align*}
    -\eta^{-1} \sum_{i \in [n]} \vr_i \log(\va_i/\vr_i) &= \eta^{-1} \sum_{i \in [n]} \vr_i(-\log(1-(1-\va_i/\vr_i))) \\ &\ge \eta^{-1} \sum_{i \in [n]} \vr_i \cdot ((1-\va_i/\vr_i) + \min\{(1-\va_i/\vr_i)^2/3, 1/10)\} \\
    &= \eta^{-1} \sum_{i \in [n]} \vr_i \cdot \min\{(1-\va_i/\vr_i)^2/3, 1/10\},
\end{align*}
because $\|\va\|_1 = \sum_{i,j} \mX_{ij} \le \|\vr\|_1$ by Lemma \ref{lemma:valid}. If any of the $\min$'s in the previous expression evaluate to $1/10$, then the expression is clearly at least $\eta^{-1}\vr_i/10 \ge \eta^{-1}/10 \cdot \mu^{-1}$, because $\mu\vr$ is integral. Otherwise, by the Cauchy-Schwarz inequality,
\begin{align*}
    \eta^{-1} \sum_{i \in [n]} \vr_i (1-\va_i/\vr_i)^2/3 = \eta^{-1}/3 \cdot \sum_{i \in [n]} (\va_i - \vr_i)^2/\vr_i \ge \eta^{-1}/3 \cdot \frac{\|\va-\vr\|_1^2}{\|\vr\|_1},
\end{align*}
as desired. This completes the proof.
\end{proof}
We can now bound the total number of iterations of the algorithm.
\begin{lemma}[Iteration count]
\label{lemma:itercount}
For integral vectors $\vr, \vc \in \mathbb{Z}^n$, and $\mu := \max\{\|\vr\|_\infty, \|\vc\|_\infty\}$ an execution of Algorithm \ref{algo:scaling} uses at most $O(\|\vr\|_1^2 \log(n\mu) \log(\eps^{-1}\|\mQ\|_\infty\mu))$ iterations.
\end{lemma}
\begin{proof}
After doubling $\eta$ the duality gap is at most $4\eta^{-1}\|\vr\|_1 \log(n\mu)$ by Lemma \ref{lemma:dualitygap}. If $\|\va - \vr\|_1 \ge 1/(2\mu)$, then the dual increase is at least $\eta^{-1}/10 \cdot \min\{\mu^{-1}, \|\vr\|_1^{-1}\|\va - \vr\|_1^2\} \ge 1/40 \cdot \eta^{-1}\|\vr\|_1^{-1}\mu^{-2}$. Hence the number of iterations during a doubling phase is bounded by $\frac{4\eta^{-1}\|\vr\|_1 \log(n\mu)}{1/40 \cdot \eta^{-1}\|\vr\|_1^{-1}\mu^{-2}} = O((\mu\|\vr\|_1)^2 \log(n\mu))$. Additionally, the total number of doubling phases is bounded by $\log((4\mu\eps^{-1}\|\vr\|_1\log(n\mu))/(10\|\mQ\|_{\infty}^{-1} \log(n\mu)))$.
Thus, the lemma follows (recall that the $\vr$ in the Lemma statement is really $\mu\vr$ after scaling).
\end{proof}
Finally, we show how to recover a feasible solution from $\mX$, and complete the proof of Theorem \ref{theo:algo}.
\begin{proof}[Proof of Theorem \ref{theo:algo}]
The iteration complexity bound follows from Lemma \ref{lemma:itercount}, so it suffices to explain how to round our final solution $\mX$ to an accurate solution $\mY$.

To construct $\mY$, let $\hat{\mX}$ be as in Lemma \ref{lemma:hatmx}, and let $\mY = 2\hat{\mX}$. By definition, we know that $\mY\vone = 2\hat{\mX}\vone = \vr$, and similarly $\mY^\top\vone = \vc$. To bound the optimality gap of $\mY$, note by the equations in the proof of Lemma \ref{lemma:dualitygap} that $-\|\vr\|_1 \log(n\mu) \le \eta(D/2 - \langle \hat{\mX}, \mQ \rangle)$, so \[ \langle \hat{\mX}, \mQ \rangle \le \eta^{-1}\|\vr\|_1 \log(n\mu) + D/2 \le \eta^{-1}\|\vr\|_1 \log(n\mu) + OPT/2, \] as $D \le OPT$ by Corollary \ref{cor:valid}.
Hence \[ \langle \mY, \mQ \rangle = 2\langle \hat{\mX}, \mQ\rangle \le 2\eta^{-1}\|\vr\|_1 \log(n \mu) + OPT \le OPT + \mu^{-1}\epsilon \] by the ending choice of $\eta$. Because Algorithm \ref{algo:scaling} scales everything down by $\mu$, the error in terms of the original objective is $\epsilon$, as desired. $\mY$ can be computed efficiently by calling maximum flow.
\end{proof}

\input{sec3_old.tex}

\printbibliography

\end{document}

%% file: sec3_old.tex
\section{Reducing to Polynomially Bounded Instances via Scaling}


\label{sec:polyScaling}

In this section, we will present cost and capacity scaling procedures that reduce solving integral OT to instances with polynomially bounded entries and prove \Cref{theo:costCapScaling}.


The following proof can be extended to the case where $n \neq m$ in the OT problem to obtain a $\poly(n,m,\log \frac{1}{\epsilon})$ time algorithm. However, one may find such a proof confusing to read, since $m$, in addition to being the size of the demand vector, also denotes the number of edges in a min-cost circulation instance. Thus, for ease of exposition, we present the proof for $n = m$.

Instead of OT, we consider the problem of finding minimum cost circulation  (MCC) on directed graphs. In the problem of minimum cost circulation, we are given a directed graph $G=(V, E)$ with integral edge costs $\bc \in \pm [C]^E$ and integral capacities $\bu \in [U]^E.$ The goal is to find a circulation $\bf$ viewed as a vector over the set of edges $E$ of minimum cost.
It is formulated as the following linear program:
\begin{align}
\label{eq:MCC}
\min_{\mB^\top \bf = 0, 0 \le \bf \le \bu} \bc^\top \bf
\end{align}
where $\mB$ is the edge-vertex incidence matrix of $G.$
We use $T_{MCC}(n, m, C, U)$ to be the time to find an integral solution that minimizes \eqref{eq:MCC}, given a graph with $n$ vertices and $m$ edges.
We also define $T_{OT}(n, C, U)$ to denote the time for solving \eqref{eq:OT} for $n$-dimensional $\vr, \vc$ within $1/\poly(n)$-additive error where $C$ is the maximum absolute value of costs and $U$ is the maximum demand or supply entries.

We first show that OT can be reduced to MCC.
\begin{lemma}
\label{lemma:OTtoMCC}
Given an integral instance of \eqref{eq:OT}, we have
\begin{align*}
T_{OT}(n, \norm{\mQ}_{\infty}, \mu) = O(n^2) + T_{MCC}(2n, n^2, \norm{\mQ}_{\infty}, \mu).
\end{align*}
\end{lemma}
\begin{proof}
First, we can construct in $O(n^2)$ time a integral matrix $\mX^{(0)}$ such that $\mX^{(0)} \mathbf{1} = \br$ and $\mX^{(0) \top} \mathbf{1} = \bc$.
Solving \eqref{eq:OT} is equivalent to finding $\bDelta$ that minimizes
\begin{align*}
    \min_{\bDelta} \sum_{i, j} \mQ_{ij} \bDelta_{ij}, \text{ such that } \mX^{(0)} + \bDelta \ge 0, \bDelta \mathbf{1} = 0, \text{and} \bDelta^\top \mathbf{1} = 0
\end{align*}
This corresponds to an MCC problem on a complete bipartite graph with $n$ vertices on each side.
The direction and capacity of each edge between the $i$-th vertex on the left and the $j$-th vertex on the right depend on the value of $\mX^{(0)}_{ij}.$
\end{proof}

Next, we show that one can reduce solving \eqref{eq:MCC} to few instances where the largest cost in absolute value is $O(n).$
This is done via a revisit of the cost scaling scheme that appears in \cite{CKLPPGS22}.
\begin{lemma}[Cost Scaling, Lemma C.3~\cite{CKLPPGS22}]
\label{lemma:costScaling}
We have
\begin{align*}
T_{MCC}(n, m, C, U) = O((T_{MCC}(n, m, 10n, U) + m)\log C)
\end{align*}
\end{lemma}
\begin{proof}
In Lemma C.8 of \cite{CKLPPGS22}, we only need the rounded cost differs from the real cost by at most $\eps / 2.$
Therefore, we only need to round edge costs to the nearest integral multiple of $\eps / 2$ within the range $[-\eps, \eps n].$
Thus, the new rounded costs are within $\pm (\eps / 2) \cdot [10n].$
\end{proof}

Given the largest cost in absolute value is $O(n)$, we can further reduce \eqref{eq:MCC} to few instances whose capacity is $\poly(n).$
This is also done via a revisit of the capacity scaling scheme of \cite{CKLPPGS22}.
\begin{lemma}[Capacity Scaling, Lemma C.10~\cite{CKLPPGS22}]
\label{lemma:capScaling}
We have
\begin{align*}
T_{MCC}(n, m, 10n, U) = O((T_{MCC}(n, m, O(n), O(m^2n^4)) + m)\log U)
\end{align*}
\end{lemma}
\begin{proof}
In Lemma C.11 of \cite{CKLPPGS22}, the cycle found via solving unit-capacitated MCC has an approximation ratio $10 m n^2$ instead of $m^{12}$ because the cost is bounded by $10n$ instead of $m^{10}.$
Thus, the rounded capacities are integers at most $O((mn^2)^2).$
\end{proof}

Finally, we show that MCC can be solved using the \textsc{Sinkhorn} algorithm with regularization scheduling.
In particular, we reduce any integral MCC to an integral OT instance.
Using the algorithm from \cref{theo:algo}, we can compute a feasible solution within $OPT + 1 / \poly(n)$.
Then, we can round the solution to a feasible integral solution without increasing the cost in $n^2$-time via a cycle cancellation procedure from \cite{KP15}.
The reduction is summarized as follows:
\begin{lemma}[Solving MCC via OT]
\label{lemma:MCCtoOT}
We have
\begin{align*}
T_{MCC}(n, m, C, U) = T_{OT}(\max\{n, m\}, mUC, mU).
\end{align*}
In addition, the total demand/supply of the reduced OT instance is $mU$ as well.
\end{lemma}
\begin{proof}
Given an instance of \eqref{eq:MCC}, we construct an integral OT instance as follows:
We define the row and column space indexed by $V$ and $E$ respectively.
For any $u \in V$, we define its demand $\br_u$ to be the weighted incoming degree $\br_u = \deg^{in}(u) = \sum_{e = (u, v)} \bu(e).$
For any edge $e \in E$, we define its supply $\bc_e$ to be its capacity $\bc_e = \bu(e).$
Clearly, both the demand and supply vectors $\br$ and $\bc$ are integers at most $m \cdot U.$
The cost matrix $\mQ \in \R^{V \times E}$ is defined as follows:
\begin{align*}
\mQ_{ue} = \begin{cases}
    \bc(e) &\text{ if } e = (u, v) \\
    0 &\text{ if } e = (v, u) \\
    m \cdot U \cdot C &\text{otherwise}
\end{cases}
\end{align*}

Next, we show that solving the OT w.r.t. $\br, \bc$, and $\mQ$ we construct is equivalent to solving the given MCC instance.
Given any integral OT solution $\mX$, we define the flow $\bf$ as follows:
\begin{align*}
    \bf_e = \mX_{ue} \ge 0, \forall e = (u, v)
\end{align*}
We have $\bc^\top \bf = \sum_{u, e} \mQ_{ue} \mX_{ue}.$
To see that $\bf$ is a circulation, let us look at the net flow at any vertex $u$
\begin{align*}
\bf^{net}(u)
&= \sum_{e = (v, u)} \bf_e - \sum_{e = (u, v)} \bf_e \\
&= \sum_{e = (v, u)} \mX_{ve} - \sum_{e = (u, v)} \mX_{ue} \\
&= \sum_{e = (v, u)} \left(\bu(e) - \mX_{ue}\right) - \sum_{e = (u, v)} \mX_{ue} \\
&= \deg^{in}(u) - \sum_{e: u \in e} \mX_{ue} = \deg^{in}(u) - \br_u = 0
\end{align*}
where the $3_{rd}$ equality comes from that the supply on edge $e$ in the OT instance is exactly $\bu(e)$, i.e. $\mX_{ue} + \mX_{ve} = \bc_e = \bu(e).$
In addition, $\mX_{ue} = 0$ whenever $\mQ_{ue} = mUC$ because $\mX$ is an optimal solution.

On the other hand, given any feasible circulation $\bf$ to the MCC instance, we can construct $\mX$, a feasible OT solution of identical cost as follows:
\begin{align*}
\mX_{ue} = \begin{cases}
    \bf_e &\text{ if } e = (u, v) \\
    \bu(e) - \bf_e &\text{ if } e = (v, u) \\
    0 &\text{otherwise}
\end{cases}
\end{align*}
Using a similar argument as above, we know that $\bc^\top \bf = \sum_{u, e} \mQ_{ue} \mX_{ue}$, $\mX \mathbf{1} = \br$, and $\mX^\top \mathbf{1} = \bc.$

Thus, to solve the MCC, we can apply \Cref{theo:algo} to solve the OT instance with $1/\poly(n)$-additive error in
\begin{align*}
\O\left(\underset{\text{\# of iterations}}{\left(\sum_u \bd^H(u)\right)^2} \cdot \underset{\text{cost per iteration}}{mn}\right) = \O((Um)^2 mn)\text{-time.}
\end{align*}
Then, we round the fractional solution to an integral one without additional error in $O(m^2)$-time (see Section 5 of \cite{KP15}).
Integrity ensures that any integral solution within $OPT+1/\poly(n)$ is an exact optimal solution.
\end{proof}

Given all these Lemmas, we are now ready to prove \Cref{theo:costCapScaling}.
\begin{proof}[Proof of \Cref{theo:costCapScaling}]
Given an integral OT instance, combining \Cref{lemma:OTtoMCC}, \Cref{lemma:costScaling}, \Cref{lemma:capScaling}, and \Cref{lemma:MCCtoOT} solves the instance in time
\begin{align*}
T_{OT}(n, \norm{\mQ}_{\infty}, \mu)
&\underset{\text{\Cref{lemma:OTtoMCC}}}{=} O(n^2) + T_{MCC}(2n, n^2, \norm{\mQ}_{\infty}, \mu) \\
&\underset{\text{\Cref{lemma:costScaling}}}{=} O\left(n^2 + T_{MCC}\left(2n, n^2, O(n), \mu\right)\log\left(\norm{\mQ}_{\infty}\right)\right) \\
&\underset{\text{\Cref{lemma:capScaling}}}{=} O\left(n^2 + T_{MCC}\left(2n, n^2, O(n), O(n^8)\right)\log\left(\norm{\mQ}_{\infty}\right)\log(\mu)\right) \\
&\underset{\text{\Cref{lemma:MCCtoOT}}}{=} O\left(n^2 + n^4 + T_{OT}\left(n^2, O(n^{11}), O(n^{10})\right)\log\left(\norm{\mQ}_{\infty}\right)\log(\mu)\right).
\end{align*}
This concludes the proof.
\end{proof}